\documentclass[a4paper]{article}
\usepackage{ifthen}
\usepackage{microtype}
\usepackage{boxedminipage}
\usepackage{amsmath}
\usepackage{amsfonts}
\usepackage{amssymb}
\usepackage{amsthm}
\usepackage{authblk}

\usepackage{amssymb,amsfonts,amsthm,epsfig}
\newtheorem{theorem}{Theorem}
\newtheorem*{thm}{Theorem}

\newtheorem{lemma}[theorem]{Lemma}
\newtheorem*{lm}{Lemma}

\newtheorem{corollary}[theorem]{Corollary}

\newtheorem{definition}[theorem]{Definition}

\newcommand{\eps}{{\epsilon}}
\newcommand{\C}[1]{{{C}^{(#1)}}}
\newcommand{\itS}[1]{{{S}^{(#1)}}}
\newcommand{\itT}[1]{{{T}^{(#1)}}}

\newcommand{\commented}{no}

\ifthenelse{\equal{\commented}{yes}}{
\newcommand{\rnote}[1]{\footnote{{\bf [[Ragesh: {#1}\bf ]] }}}

\newcommand{\snote}[1]{\footnote{{\bf [[Sandeep: {#1}\bf ]] }}}
}

\ifthenelse{\equal{\commented}{no}}{
\newcommand{\jnote}[1]{}
\newcommand{\snote}[1]{}
\newcommand{\rnote}[1]{}
}

\title{A simple $D^2$-sampling based PTAS for $k$-means and other Clustering problems}

\author[1]{Ragesh Jaiswal}
\author[1]{Amit Kumar}
\author[1]{Sandeep Sen} 
\affil[1]{Department of Computer Science and Engineering \\ IIT Delhi\\ \texttt{\{rjaiswal,amitk,ssen\}@cse.iitd.ac.in}}

\begin{document}

\maketitle

\begin{abstract}
Given a set of points $P \subset \mathbb{R}^d$, the $k$-means clustering problem is to find a set of $k$ {\em centers} $C = \{c_1,...,c_k\}, c_i \in \mathbb{R}^d,$ such that the  objective function $\sum_{x \in P} d(x,C)^2$, where $d(x,C)$ denotes the distance between $x$ and the closest center in $C$, is
minimized. This is one of the most prominent objective functions that have been studied with respect to clustering.

$D^2$-sampling \cite{ArthurV07} is a simple non-uniform sampling technique for choosing points from a set of points. It works as follows: given a set of points $P \subseteq \mathbb{R}^d$, the first point is chosen uniformly at random from $P$. Subsequently, a point from $P$ is chosen as the next sample with probability proportional to the square of the distance of this point to the nearest previously sampled points.

$D^2$-sampling has been shown to have nice properties with respect to the $k$-means clustering problem. Arthur and Vassilvitskii \cite{ArthurV07} show that $k$ points chosen as centers from $P$ using $D^2$-sampling gives an $O(\log{k})$ approximation in expectation. Ailon et. al. \cite{AJMonteleoni09} and Aggarwal et. al. \cite{AggarwalDK09} extended results of \cite{ArthurV07} to show that $O(k)$ points chosen as centers using $D^2$-sampling give $O(1)$ approximation to the $k$-means objective function with high probability. In this paper, we further demonstrate the power of $D^2$-sampling by giving a simple randomized $(1 + \epsilon)$-approximation algorithm that uses
the $D^2$-sampling in its core.
\end{abstract}

\section{Introduction}
Clustering problems arise in diverse areas including  machine learning, data mining, image processing and
web-search \cite{broder97,faloutsos,deer90,Swain}.
One of the most commonly used clustering problems is the $k$-means problem. Here,
we are given a set of points $P$ in a $d$-dimensional Euclidean space, and a parameter $k$. The goal is
to find a set $C$ of $k$ {\em centers} such that the objective function $$\Delta(P, C) = \sum_{p \in P} d(p,C)^2$$ is minimized,
where $d(p,C)$ denotes the distance from $p$ to the closest center in $C$. This naturally partitions $P$ into
$k$ clusters, where each cluster corresponds to the set of points of $P$ which are closer to a particular center
than other centers. It is also easy to show that the center of any cluster must be  the mean of the
points in it.
 In most applications, the parameter $k$ is a small constant. However, this problem turns out to be
 NP-hard even for $k=2$ \cite{das08}.

One very popular heuristic for solving the $k$-means problem is the Lloyd's algorithm~\cite{lloyd}. The heuristic is
as follows : start with an arbitrary set of $k$ centers as seeds. Based on these $k$ centers, partition the
set of points into $k$ clusters, where each point gets assigned to the closest center. Now, we update the
set of centers as the means of each of these clusters. This process is repeated till we get convergence. Although,
this heuristic often performs well in practice, it is known that it can get stuck in local minima~\cite{ArthurV06}.
There has been lot of recent research in understanding why this heuristic works fast in practice, and how it can be modified such that we can guarantee that the solution produced by this heuristic is always close to the optimal solution.

One such modification  is to carefully choose the set of initial $k$ centers.
Ideally, we would like to pick these centers such that we have a center close to each of the optimal clusters.
Since we do not know the optimal clustering, we would like to make sure that these centers are well separated from each
other and yet, are representatives of the set of points. A recently proposed idea~\cite{OstrovskyRSS06,ArthurV07} is to
pick the initial centers using $D^2$-sampling which can be described as follows.
The first center is picked uniformly at random from the set of
points $P$. Suppose we have picked a set of $k' < k$ centers -- call this set $C'$.
Then a point $p \in P$ is chosen as the next center
with probability proportional to $d(p, C')^2$. This process is repeated till we have a set of $k$ centers.

There has been lot of recent activity in understanding how good a set of centers picked by $D^2$-sampling are (even
if we do not run the Lloyd's algorithm on these seed centers).
Arthur and  Vassilvitskii~\cite{ArthurV07} showed that if we pick $k$ centers with $D^2$-sampling,
then the expected cost of the corresponding solution to the $k$-means instance is within $O(\log k)$-factor of
the optimal value. Ostrovsky et. al.~\cite{OstrovskyRSS06} showed that if the set of points satisfied a separation 
condition (named $(\eps^2,k)$-irreducible as defined in Section~\ref{sec:pre}), then these $k$ centers give a constant factor approximation for the $k$-means problem. Ailon et. al.~\cite{AJMonteleoni09}
proved a bi-criteria approximation property -- if we pick $O(k \log k)$ centers by $D^2$-sampling, then it is
a constant approximation, where we compare with the optimal solution that is allowed to pick $k$ centers only.
Aggarwal et. al.~\cite{AggarwalDK09} give an improved result and show that it is enough to pick $O(k)$ centers by
$D^2$-sampling to get a constant factor bi-criteria approximation algorithm.

In this paper, we give yet another illustration of the power of the $D^2$-sampling idea. We give a simple randomized
$(1+\eps)$-approximation algorithm for the $k$-means algorithm, where $\eps > 0$ is an arbitrarily small constant.
At the heart of our algorithm is the idea of $D^2$-sampling -- given a set of already selected centers, we pick
a small set of points by $D^2$-sampling with respect to these selected centers. Then, we pick the next
center as the centroid of a subset of these small set of points. By repeating this process of picking $k$ centers
sufficiently many times, we can guarantee that with high probability, we will get a set of $k$ centers whose objective
value is close to the optimal value. Further, the running time of our algorithm is $O(nd \cdot 2^{\tilde{O}(k^2/\eps)})$ \footnote{ $\tilde O$
notation hides a $O(\log k/\eps)$ factor which simplifies the expression.}--
for constant value of $k$, this is a linear time algorithm.
 It is important to note that PTAS with better running time are known
for this problem. Chen \cite{Chen06} give an $O\left(nkd + d^2 n^{\sigma} \cdot 2^{(k/\epsilon)^{O(1)}}\right)$ algorithm for any $\sigma > 0$ and Feldman et al. \cite{FeldmanMS07} give an $O\left(nkd + d \cdot poly(k/\epsilon) + 2^{\tilde{O}(k/\epsilon)}\right)$ algorithm.
However, these results often are quite involved, and use the notion of coresets. Our algorithm is simple, and only uses the concept of $D^2$-sampling.

\subsection{Other Related Work}
There has been significant research on exactly solving the $k$-means algorithm (see e.g.,~\cite{inaba}), but all of these algorithms
take $\Omega(n^{kd})$ time. Hence, recent research on this problem has focused on obtaining fast $(1+\eps)$-approximation
algorithms for any $\eps > 0$. Matousek~\cite{Matousek00} gave  a PTAS
with running time  $O(n \eps^{-2k^2d} \log^k n)$. Badoiu et al.~\cite{BadoiuHI02}
gave an improved PTAS with running time
$O(2^{(k/\eps)^{O(1)}} d^{O(1)}n\log^{O(k)}n)$. de la Vega et al.~\cite{VegaKKR03} gave a PTAS which works well for points
in high dimensions. The running time of this algorithm is $O(g(k,\eps)n \log^kn)$ where
$g(k,\eps) = \exp[(k^3/\eps^8)(\ln (k/\eps) \ln k]$.  Har-Peled et al.~\cite{Har-PeledM04}  proposed a PTAS whose running time
is $O(n+k^{k+2}\eps^{-(2d+1)k} \log^{k+1}n \log^k\frac{1}{\eps})$. Kumar et al.~\cite{KumarSS10} gave the first
linear time PTAS for fixed $k$ -- the running time of their algorithm is $O(2^{(k/\eps)^{O(1)}} dn)$. Chen~\cite{Chen06}
used the a new coreset construction to  give a PTAS with improved running time of $O(ndk + 2^{(k/\eps)^{O(1)}} d^2n^\sigma)$.
Recently, Feldman et al.~\cite{FeldmanMS07} gave a PTAS with running time $O(nkd + d  \cdot poly(k/\eps) +
2^{{\tilde O} (k/\eps)})$ -- this is the fastest known PTAS (for fixed $k$) for this problem.

There has also been work on obtaining fast constant factor approximation algorithms for the $k$-means problem
based on some properties of the input points (see e.g.~\cite{OstrovskyRSS06,AwasthiBS10}).

\subsection{Our Contributions}
In this paper, we give a simple PTAS for the $k$-means problem based on the idea of $D^2$-sampling. Our work
builds on and simplifies the result of Kumar et al.~\cite{KumarSS10}. We briefly describe their algorithm first.
It is well known that for the 1-mean problem, if we sample a set of $O(1/\eps)$ points uniformly at random, then
the mean of this set of sampled points is close to the overall mean of the set of all points. Their algorithm
begins by sampling $O(k/\eps)$ points uniformly at random. With reasonable probability, we would sample $O(1/\eps)$
points from the largest cluster, and hence we could get a good approximation to the center corresponding to this cluster
(their algorithm tries all subsets of size $O(1/\eps)$ from the randomly sampled points). However, the other clusters
may be much smaller, and we may not have sampled enough points from them. So, they need to prune a lot of points
from the largest cluster so that in the next iteration a random sample of $O(k/\eps)$ points will contain $O(1/\eps)$
points from the second largest cluster, and so on. This requires a  non-trivial idea termed as {\em tightness}
 condition by the authors. In this paper,
we show that the pruning is not necessary if  instead of using uniform random sampling, one uses $D^2$-sampling.

We can informally describe our algorithm as follows. We maintain a set of
candidate centers $C$, which is initially empty.
Given a set $C$, $|C| < k$, we add a new center to $C$ as follows. We sample a set $S$ of $O(k/\eps^3)$ points using
$D^2$-sampling with respect to $C$. From this set of sampled points, we pick a
subset $T$ and the new center is
the mean of this set $T$. We add this to $C$ and continue.

From the property of $D^2$-sampling (\cite{AggarwalDK09,AJMonteleoni09}),
with some constant, albeit small probability $p'$, we pick up a point from
a hitherto untouched cluster $C'$ of the {\em optimal clustering}. Therefore
by sampling about $\alpha /p'$ points using $D^2$-sampling, we expect to hit
approximately $\alpha $ points from $C'$. If $\alpha $ is large enough,
(c.f. Lemma~\ref{lem:inaba}), then the centroid of these $\alpha$ points gives a
$(1 +\eps)$ approximation of the cluster $C'$.
Therefore, with reasonable probability, there will be
a choice of a subset $T$ in each iteration such that the set of
centers chosen are from $C'$. Since we do not know $T$, our algorithm will
try out all subsets of size $|T|$ from the sample $S$.
Note that our algorithm is very
simple, and can be easily parallelized.
Our algorithm has running time $O(dn \cdot 2^{\tilde{O}(k^2/\eps)})$ which is an improvement over that of Kumar et al.~\cite{KumarSS10} who
gave a PTAS with running time $O\left(nd \cdot 2^{(k/\eps)^{O(1)}}\right)$.
\footnote{It can be used in conjunction with Chen~\cite{Chen06} to obtain
a superior running time but at the cost of the simplicity of our approach}

Because of the relative simplicity, our algorithm generalizes to measures
like Mahalanobis distance and $\mu$-similar Bregman divergence. Note that
these do not satisfy triangle inequality and therefore not strict metrics.
Ackermann et al. \cite{ab09} have generalized the framework of Kumar et al.
\cite{KumarSS10} to Bregman divergences but we feel that the $D^2$-sampling based
algorithms are simpler.

We formally define the problem and give some preliminary results in Section~\ref{sec:pre}.
In Section~\ref{sec:algo}, we describe our algorithm, and then analyze it subsequently. In Section~\ref{sec:other}, we discuss PTAS for other distance measures.

\section{Preliminaries}
\label{sec:pre}
An instance of the $k$-means problem consists of a set $P \subseteq \mathbb{R}^d$ of $n$ points in $d$-dimensional space and a parameter $k$.
For a set of points (called centers) $C \subseteq \mathbb{R}^d$, let $\Delta(P,C)$ denote $\sum_{p \in P} d(p,C)^2, $ i.e., the cost of the
solution which picks $C$ as the set of centers. For a singleton $C = \{c\}$, we shall often abuse notation,
and use $\Delta(P, c)$ to denote $\Delta(P, C)$.
 Let $\Delta_k(P)$ denote the cost of the optimal $k$-means
solution for $P$.

\begin{definition}
Given a set of points $P$ and a set of centers $C$, a point $p \in P$ is said to be sampled using {\em $D^2$-sampling}
with respect to $C$ if the probability of it being sampled, $\rho(p)$,  is given by
$$\rho(p) = \frac{d(p,C)^2}{\sum_{x \in P} d(x,C)^2} =
\frac{\Delta(\{p\}, C)}{\Delta(P,C)}. $$
\end{definition}

We will also need the following definition from~\cite{KumarSS10}.
\begin{definition}[Irreducibility or separation condition] Given $k$ and $\eps$, a set of points $P$ is said to be $(k,\gamma)$-irreducible if
$$\Delta_{k-1}(P) \geq (1+\gamma) \cdot \Delta_{k}(P). $$
\end{definition}

We will often appeal to the following result~\cite{inaba} which shows that uniform random sampling works well for $1$-means\footnote{It turns out
that even minor perturbations from uniform distribution can be catastrophic and indeed in this paper we had to work around this.}.
\begin{lemma}[Inaba et al.~\cite{inaba}]
\label{lem:inaba}
Let $S$ be a set of points obtained by independently sampling $M$ points with replacement uniformly at random from a point set $P$. Then, for any $\delta > 0$,
\[
\Delta(P, \{m(S)\}) \leq \left(1 +  \frac{1}{\delta M}\right) \cdot \Delta(P, \{m(P)\}),
\]
holds with probability at least $(1 - \delta)$. 
Here $m(X) =  \left(\frac{\sum_{x \in X} x}{|X|}\right)$ denotes the centroid of a point set $X$.
\end{lemma}

Finally, we will use the following property of the squared Euclidean metric. This is a standard result from linear algebra \cite{ps05}.

\begin{lemma}\label{lemma:2}
Let $P \subseteq \mathbb{R}^d$ be any point set and let $c \in \mathbb{R}^d$ be any point. Then we have the following:
\[
\sum_{p \in P} d(p, c)^2 = \sum_{p \in P} d(p, m(P))^2 + |P| \cdot d(c, m(P))^2,
\]
where $m(P) =  \left(\frac{\sum_{p \in P} p}{|P|}\right)$ denotes the centroid of the point set.
\end{lemma}

Finally, we mention the simple approximate triangle inequality with respect to the squared Euclidean distance measure.

\begin{lemma}[Approximate triangle inequality]\label{lemma:3}
For any three points $p, q, r \in \mathbb{R}^d$ we have:
\[
d(p, q)^2 \leq 2 \cdot (d(p, r)^2 + d(r, q)^2).
\]
\end{lemma}

\section{PTAS for $k$-means}
\label{sec:algo}
We first give a high level description behind the  algorithm. We will also assume that the instance is $(k, \eps)$-irreducible for a
suitably small parameter $\eps$. We shall then get rid of this assumption later. The algorithm is described
in Figure~\ref{fig:k}. Essentially, the algorithm maintains a set $C$ of centers,
where $|C| \leq k$. Initially $C$ is empty, and in each iteration of Step 2(b), it adds one center to $C$ till
its size reaches $k$. Given a set $C$, it samples  a set of $S$ points from $P$ using $D^2$-sampling with respect
to $C$ (in Step 2(b)). Then it picks  a subset $T$ of $S$ of size $M = O(1/\eps)$,
and adds the centroid of $T$ to $C$. The algorithm
cycles through all possible subsets of size $M$ of $S$ as choices for $T$, and for each such choice,
repeats the above steps to find the next center, and so on.  To make the presentation clearer,
we pick a $k$-tuple of $M$-size subsets $(s_1, \ldots, s_k)$ in advance, and when $|C|=i$, we pick $T$ as the $s_i^{th}$ subset of $S$. 
In Step 2(i), we cycle through all such $k$-tuples $(s_1, \ldots, s_k)$. In the analysis, we just need
to show that {\em one} such $k$-tuple works with reasonable probability.

We develop some notation first. For the rest of the analysis, we will fix a tuple $(s_1, \ldots, s_k)$ --
this will be the ``desired tuple'', i.e., the one for which we can show that the set $C$ gives a good solution.
As our analysis proceeds, we will argue what properties this tuple should have. Let $\C{i}$ be the set $C$
at the beginning of the $i^{th}$ iteration of Step 2(b). To begin with $\C{0}$ is empty. Let $\itS{i}$
be the set $S$ sampled during the $i^{th}$ iteration of Step 2(b), and $\itT{i}$ be the corresponding set $T$
(which is the $s_i^{th}$ subset of $\itS{i}$).

Let $O_1, \ldots, O_k$ be the optimal clusters, and $c_i$ denote the centroid  of points in $O_i$. Further, let $m_i$ denote $|O_i|$,
and wlog assume that $m_1 \geq \ldots \geq m_k$. Note that $\Delta_1(O_i)$ is same as
$\Delta(O_i, \{c_i\})$.
 Let $r_i$ denote the average cost paid by a point in $O_i$, i.e.,
$$r_i = \frac{\sum_{p \in O_i} d(p,c_i)^2}{m_i}. $$
We will assume that the input set of points $P$ are
$(k, \eps)$-irreducible. We shall remove this assumption later. Now we show that any two optimal centers are far enough.

\begin{center}
\begin{figure}
\begin{boxedminipage}{5in}
{\bf Find-k-means(P)}

\hspace{0.1in} Let $N = (51200 \cdot k/\eps^3)$, $M = 100/\eps$, and $P = \binom{N}{M}$

\hspace{0.1in} 1. {\bf Repeat} $2^k$ times and output the the set of centers $C$ that give least cost

\hspace{0.3in} 2. {\bf Repeat} for all $k$-tuples $(s_1, ..., s_k) \in [P]\times [P] \times .... \times [P]$ and

\hspace{0.3in} \ \ \ \ \ pick the set of centers $C$ that gives least cost

\hspace{0.5in} \ \ \ (a) $C \leftarrow \{\}$

\hspace{0.5in} \ \ \ (b) For $i \leftarrow 1$ to $k$

\hspace{0.7in}\ \ \ \ \  \ \ Sample a set $S$ of $N$ points with $D^2$-sampling (w.r.t. centers $C$)

\hspace{0.7in} \ \ \ \ \ \ \ Let $T$ be the $s_i^{th}$ subset of $S$. \footnote{For a set of size $N$ we consider an arbitrary ordering of the subsets of size $M$ of this set.}

\hspace{0.7in} \ \ \ \ \ \ \ $C \leftarrow C \cup \{m(T)\}$. \footnote{$m(T)$ denote the centroid of the points in $T$.}

\end{boxedminipage}
\caption{The $k$-means algorithm that gives $(1+\eps)$-approximation for any $(k, \eps)$-irreducible data set. Note that the inner loop is executed at most 
$2^k \cdot \left( \binom{N}{M} \right)^k \sim 2^k \cdot 2^{{\tilde O} (k/\eps)}$ times. }
\label{fig:k}
\end{figure}
\end{center}
\vspace*{-0.5in}

\begin{lemma}
\label{lem:min}
For any $1 \leq i, j \leq k, i \neq j$,
$$ d(c_i, c_j)^2 \geq \eps \cdot (r_i + r_j). $$
\end{lemma}
\begin{proof}
Suppose $i > j$, and hence $m_i \geq m_j$.
For the sake of contradiction assume $d(c_i, c_j)^2 <  \eps \cdot (r_i + r_j)$. Then we have,
\begin{eqnarray*}
\Delta(O_i \cup O_j, \{c_i\}) & = & m_i \cdot r_i + m_j \cdot r_j + m_j \cdot d(c_i, c_j)^2  \quad \textrm{(using Lemma~\ref{lemma:2})}\\
&\leq&  m_i \cdot r_i + m_j \cdot r_j + m_j \cdot \eps \cdot (r_i + r_j)\\
&\leq& (1 + \eps) \cdot m_i \cdot r_i +(1 + \eps) \cdot  m_j \cdot r_j \quad \textrm{(since $m_i \geq m_j$)}\\
&\leq& (1 + \eps) \cdot \Delta(O_i \cup O_j, \{c_i, c_j\})
\end{eqnarray*}
This implies that the centers $\{c_1, ..., c_k\} \backslash \{c_j\}$ give a $(1 + \eps)$-approximation to the
$k$-means objective. This contradicts the assumption that $P$ is $(\eps, k)$-irreducible.
\end{proof}

We give an outline of the proof. Suppose in the first $i-1$ iterations, we have found centers which are close to the centers of some 
$i-1$ clusters in the optimal solution. Conditioned on this fact, we show that in the next iteration, we are likely to sample enough 
number of points from one of the remaining clusters (c.f. Corollary~\ref{cor:sample}). Further, we show that the samples from this new
cluster are close to uniform distribution (c.f. Lemma~\ref{lem:key}). Since such a sample does not come from exactly uniform distribution, we
cannot apply Lemma~\ref{lem:inaba} directly. In fact, dealing with the slight non-uniformity turns out to be non-trivial (c.f. Lemmas~\ref{lem:clinaba} and ~\ref{lem:final}). 

We now show that the following invariant will hold for all iterations  : let $\C{i-1}$ consist of centers
$c_1', \ldots, c_{i-1}'$ (added in this order). Then, with probability at least $\frac{1}{2^{i}}$, there exist distinct
indices $j_1, \ldots, j_{i-1}$ such that
for all $l = 1, \ldots, i-1$,
\begin{equation}
\label{eq:inv}
 \Delta(O_{j_l},c_l') \leq (1 + \eps/20) \cdot \Delta(O_{j_l},c_{j_l})
\end{equation}
Suppose this invariant holds for $\C{i-1}$ (the base case is easy since $\C{0}$ is empty).
We now show that this invariant holds for $\C{i}$ as well. In other words, we just show that in the $i^{th}$ iteration,
with probability at least $1/2$,  the
algorithm finds a center $c_{i}'$  such that
$$ \Delta(O_{j_i},c_i') \leq (1 + \eps/20) \cdot \Delta(O_{j_i},c_{j_i}), $$ where $j_i$ is an index distinct from
$\{j_1, \ldots, j_{i-1}\}$.
This will basically show that at the end of the last iteration, we will have $k$ centers that give a $(1 + \eps)$-approximation
with probability at least $2^{- k}$.

We now show that the invariant holds for $\C{i}$. We use the notation developed above for $\C{i-1}$. Let $I$ denote
the set of indices $\{j_1, \ldots, j_{i-1}\}$. Now let $j_i$ be the index $j \notin I$ for which $\Delta(O_j, \C{i-1})$
is maximum. Intuitively, conditioned on sampling from clusters in $O_{i}, \cdots, O_{k}$ using $D^2$-sampling, it is  likely that
enough points from $O_{j_i}$ will be sampled.
The next lemma shows that there is good chance that elements from the sets $O_j$ for $j \notin I$
will be sampled.

\begin{lemma}
\label{lem:sampled}
 $$\frac{\sum_{l \notin I} \Delta(O_l, \C{i-1})}{\sum_{l=1}^k \Delta(O_l, \C{i-1})} \geq \eps/2.$$
\end{lemma}
\begin{proof}
Suppose, for the sake of contradiction, the above statement does not hold. Then,
\begin{eqnarray*}
\Delta(P, \C{i-1}) &=& \sum_{l \in I}  \Delta(O_l, \C{i-1}) +  \sum_{l \notin I} \Delta(O_l, \C{i-1})  \\
& < &  \sum_{l \in I} \Delta(O_l, \C{i-1})  +  \frac{\eps/2}{1 - \eps/2} \cdot \sum_{l \in I}  \Delta(O_l, \C{i-1})
\quad \textrm{(by our assumption)}\\
&=& \frac{1}{1 - \eps/2} \cdot \sum_{l \in I} \Delta(O_l, \C{i-1}) \\
&\leq& \frac{1 + \eps/20}{1 - \eps/2} \cdot \sum_{l \in I}  \Delta_1(O_l) \quad \textrm{(using the invariant for $\C{i-1}$)}\\
&\leq& (1+\eps) \cdot  \sum_{l \in I}  \Delta_1(O_l)
\ \leq \  (1+\eps) \cdot  \sum_{l \in [k]}  \Delta_1(O_l)
\end{eqnarray*}

But this contradicts the fact that $P$ is $(k, \eps)$-irreducible.
\end{proof}

\noindent
We get the following corollary easily.
\begin{corollary}
\label{cor:sample}
$$ \frac{\Delta(O_{j_i}, \C{i-1})}{\sum_{l=1}^k  \Delta(O_l, \C{i-1})} \geq \frac{\eps}{2k}.$$
\end{corollary}

The above Lemma and its Corollary say that with probability at least $\frac{\eps}{2k}$,
points in the set $O_{j_i}$ will be sampled. However the points within $O_{j_i}$ are not sampled uniformly.
Some points in $O_{j_i}$ might be sampled with higher probability than other points.
In the next lemma, we show that each point will be sampled with certain minimum probability.

\begin{lemma}
\label{lem:key}
For any $l \notin  I$ and any point $p \in O_l$, $$\frac{d(p, \C{i-1})^2}{\Delta(O_l, \C{i-1})} \geq \frac{1}{m_l} \cdot \frac{\eps}{64}. $$
\end{lemma}
\begin{proof}
Fix a point $p \in O_l$.
Let $j_t \in I$ be the
 index such that $p$ is closest to $c_{t}'$ among all centers in $\C{i-1}$.
We have
\begin{eqnarray}
\Delta(O_{l}, \C{i-1}) &\leq& m_{l} \cdot r_{l} + m_l \cdot d(c_l, c_t')^2 \quad \textrm{(using Lemma~\ref{lemma:2})}\nonumber \\
&\leq& m_{l} \cdot r_{l} + 2 \cdot m_{l} \cdot \left(d(c_l, c_{j_t})^2 + d(c_{j_t}, c_t')^2 \right) \quad \textrm{(using Lemma~\ref{lemma:3})} \nonumber \\
\label{eq:p1}
&\leq& m_l \cdot r_l + 2 \cdot m_l \cdot \left( d(c_l, c_{j_t})^2 + \frac{ \eps r_t}{20} \right),
\end{eqnarray}
where the second inequality follows from the invariant condition for $\C{i-1}$.
Also, we know that
\begin{eqnarray}
\nonumber
d(p,c_t')^2 & \geq & \frac{d(c_{j_t}, c_l)^2}{8} - d(c_{j_t}, c_t')^2 \quad
\textrm{(using Lemma~\ref{lemma:3})} \\
\nonumber
& \geq & \frac{d(c_{j_t}, c_l)^2}{8} - \frac{\eps}{20} \cdot r_t \quad \textrm{(using the invariant for $\C{i-1}$)} \\
\label{eq:p2}
&\geq& \frac{d(c_{j_t}, c_l)^2}{16} \quad \textrm{(Using Lemma~\ref{lem:min})}
\end{eqnarray}
So, we get
\begin{eqnarray*}
\frac{d(p, \C{i-1})^2}{\Delta(O_l, \C{i-1})} & \geq& \frac{d(c_{j_t}, c_l)^2}
{16 \cdot m_l \cdot \left( r_l + 2 \left( d(c_{j_t},c_l)^2 + \frac{\eps r_t}{20} \right) \right)} \quad \textrm{(using (\ref{eq:p1}) and (\ref{eq:p2}))} \\
&\geq& \frac{1}{16 \cdot m_l} \cdot \frac{1}{(1/\eps) + 2 + 1/20} \
\geq \ \frac{\eps}{64 \cdot m_l} \quad \textrm{(using Lemma~\ref{lem:min})}
\end{eqnarray*}
\end{proof}

Recall that $\itS{i}$ is the sample of size $N$ in this iteration. We would like to show that  that the invariant
will hold in this iteration as well. We first prove a simple corollary of Lemma~\ref{lem:inaba}.

\begin{lemma}
\label{lem:clinaba}
Let $Q$ be a set of $n$ points, and $\gamma$ be a parameter, $0 < \gamma < 1$. Define a random variable $X$ as follows :
with probability $\gamma$, it picks an element of $Q$ uniformly at random, and with  probability $1-\gamma$, it does not
pick any element (i.e., is null). Let $X_1, \ldots, X_\ell$ be $\ell$ independent copies of $X$, where $\ell = \frac{400}{\gamma \eps}.$
Let $T$ denote the (multi-set) of elements of $Q$ picked by $X_1, \ldots, X_\ell$. Then, with probability at least $3/4$,
$T$ contains a subset $U$ of size $\frac{100}{\eps}$ which satsifies
\begin{eqnarray}
\label{eq:clinaba}
\Delta(P, m(U)) \leq \left(1 + \frac{\eps}{20} \right) \Delta_1(P)
\end{eqnarray}
\end{lemma}
\begin{proof}
 Define a random variable $I$, which is a subset of the index set $\{1, \ldots, \ell\}$, as follows
$I = \{ t : X_t \mbox{ picks an element of $Q$, i.e., it is not null} \}$. Conditioned on $I = \{t_1, \ldots, t_r\}$, note that the random variables $X_{t_1}, \ldots,
X_{t_r}$ are independent uniform samples from $Q$. Thus if $|I| \geq \frac{100}{\eps}$, then Lemma~\ref{lem:inaba} implies that with
probability at least 0.8, the desired event~(\ref{eq:clinaba}) happens. But the expected value of $|I|$ is $\frac{400}{\eps}$, and so,
$|I| \geq \frac{100}{\eps}$ with high probability, and hence, the statement in the lemma is true.
\end{proof}

\noindent
We are now ready to prove the main lemma.
\begin{lemma}
\label{lem:final}
With probability at least $1/2$,
there exists a subset $\itT{i}$ of $\itS{i}$ of size at most $\frac{100}{\eps}$ such that
$$ \Delta(O_{j_i}, m(\itT{i})) \leq (1 + \frac{\eps}{20}) \cdot \Delta_1(O_{j_i}). $$
\end{lemma}
\begin{proof}
Recall that $\itS{i}$ contains $N = \frac{51200 k}{\eps^3}$ independent samples of $P$ (using $D^2$-sampling). We are interested in $\itS{i} \cap O_{j_i}$.
Let $Y_1, \ldots, Y_N$ be  $N$ independent random variables defined as follows : for any $t$, $1 \leq t \leq N$, $Y_t$ picks an element of $P$ using $D^2$-sampling with respect to
$\C{i-1}$. If this element is not in $O_{j_i}$, it just discards it (i.e., $Y_t$ is null). Let $\gamma$ denote $\frac{\eps^2}{128 k}$. Corollary~\ref{cor:sample}
 and Lemma~\ref{lem:key} imply that
$Y_t$ picks a particular element of $O_{j_i}$ with probability at least $\frac{\gamma}{m_{j_i}}$. We would now like to apply Lemma~\ref{lem:clinaba}
(observe that $N = \frac{400}{\gamma \eps}$). We can do this by a simple coupling argument as follows.
For a particular element $p \in O_{j_i}$,
suppose $Y_t$ assigns probability $\frac{\gamma(p)}{m_{j_i}}$ to it.
 One way of sampling a random variable
$X_t$ as in Lemma~\ref{lem:clinaba} is as follows --  first sample using $Y_t$. If $Y_t$ is null then, $X_t$ is also null. Otherwise, suppose
$Y_t$ picks an element $p$ of $O_{j_i}$. Then, $X_t$ is equal to $p$ with probability $\frac{\gamma}{\gamma(p)}$, null otherwise. It is easy
to check that with probability $\gamma$, $X_t$ is a uniform sample from $O_{j_i}$, and null with probability $1-\gamma$. Now, observe that
the set of elements of $O_{j_i}$ sampled by $Y_1, \ldots, Y_N$ is always a superset of $X_1, \ldots, X_N$. We can now use Lemma~\ref{lem:clinaba}
to finish the proof.
\end{proof}

 Thus, we will take the index $s_i$ in Step 2(i) as the index of the set $\itT{i}$ as guaranteed by the Lemma above.
Finally, by repeating the entire process $2^k$ times, we make sure that we get a $(1+\eps)$-approximate solution
with high probability.
Note that the total running time of our algorithm is
$\left( nd \cdot 2^k \cdot 2^{\tilde{O}(k/\eps)} \right)$.

\noindent
{\bf Removing the $(k,\eps)$-irreducibility assumption :} We now show how to remove this assumption. First note
that we have shown the following result.

\begin{theorem}
If a given point set $(k, \frac{\eps}{(1+\eps/2) \cdot k})$-irreducible, then
there is an algorithm that gives a $(1 + \frac{\eps}{(1+\eps/2) \cdot k})$-approximation to the $k$-means objective
and that runs in time $O(nd \cdot 2^{\tilde{O}(k^2/\eps)})$.
\end{theorem}
\begin{proof}
The proof can be obtained by replacing $\eps$ by $\frac{\eps}{(1 + \eps/2) \cdot k}$ in the above analysis. \end{proof}

Suppose the point set $P$ is not $(k,\frac{\eps}{(1+\eps/2) \cdot k})$-irreducible.
In that case it will be sufficient to find fewer centers that $(1 + \eps)$-approximate the $k$-means objective.
The next lemma shows this more formally.

\begin{theorem}
There is an algorithm that runs in time $O(nd \cdot 2^{\tilde{O}(k^2/\eps)})$ and
gives a $(1 + \eps)$-approximation to the $k$-means objective.
\end{theorem}
\begin{proof}
Let $P$ denote the set of points.  Let $1 < j \leq k$ be the largest index such that $P$ is $(i, \frac{\eps}{(1+\eps/2) \cdot k})$-irreducible.
If no such $i$ exists, then $$\Delta_1(P) \leq \left(1+\frac{\eps}{(1+\eps/2)\cdot k} \right)^k \cdot \Delta_k(P) \leq (1+\eps) \cdot \Delta_k(P), $$
and so picking the centroid of $P$ will give a $(1+\eps)$-approximation.

Suppose such an $i$ exists. In that case, we consider the $i$-means problem and from the previous lemma
we get that there is an algorithm that runs in time $O(nd \cdot 2^i \cdot 2^{\tilde{O}(i^2/\eps)})$ and gives a
$(1 + \frac{\eps}{(1+\eps/2) \cdot k})$-approximation to the $i$-means objective.
Now we have that $$\Delta_i \leq \left(1 + \frac{\eps}{(1+\eps/2) \cdot k} \right)^{k-i} \cdot \Delta_k \leq (1+\eps) \cdot \Delta_k.$$
Thus, we are done.
\end{proof}

\section{Other Distance Measures}\label{sec:other}
In the previous sections, we looked at the $k$-means problem where the dissimilarity or distance measure was the square of Euclidean distance. There are numerous practical clustering problem instances where the dissimilarity measure is not a function of the Euclidean distance. In many cases, the points are not generated from a metric space. In these cases, it makes sense to talk about the general $k$-median problem that can be defined as follows:

\begin{definition}[$k$-median with respect to a dissimilarity measure]
Given a set of $n$ objects $P \subseteq \mathcal{X}$ and a dissimilarity measure $D: \mathcal{X} \times \mathcal{X} \rightarrow \mathbb{R}_{\geq 0}$, find a subset $C$ of $k$ objects (called medians) such that the following objective function is minimized:
\[
\Delta(P, C) = \sum_{p \in P} \min_{c \in C} D(p, c)
\]
\end{definition}

In this section, we will show that our algorithm and analysis can be easily generalized and extended to dissimilarity measures that satisfy some simple properties. We will look at some interesting examples.
We start by making the observation that in the entire analysis of the previous section the only properties of the distance measure that we used were given in Lemmas~\ref{lem:inaba}, \ref{lemma:2}, and \ref{lemma:3}. We also used the symmetry property of the Euclidean metric implicitly. This motivates us to consider dissimilarity measures on spaces where these lemmas  (or mild relaxations of these) are true. For such measures, we may replace $d(p, q)^2$ (this is the square of the Euclidean distance) by $D(p, q)$ in all places in the previous section and obtain a similar result. We will now formalize these ideas.

First, we will describe a property that captures Lemma~\ref{lem:inaba}. This is similar to a definition by Ackermann et. al. \cite{abs10} who discuss PTAS for the $k$-median problem with respect to metric and non-metric distance measures.

\begin{definition}[$(f, \gamma, \delta)$-Sampling property]
Given $0 < \gamma, \delta \leq 1$ and $f: \mathbb{R} \times \mathbb{R} \rightarrow \mathbb{R}$, a distance measure $D$ over space $\mathcal{X}$ is said to have $(f, \gamma, \delta)$-sampling property if the following holds:
for any set $P \subseteq \mathcal{X}$, a uniformly random sample $S$ of $f(\gamma, \delta)$ points from $P$ satisfies
\[
Pr\left[\sum_{p \in P} D(p, m(S)) \leq (1+\gamma) \cdot \Delta_1(P)\right] \geq (1 - \delta),
\]
where $m(S) = \frac{\sum_{s \in S} s}{|S|}$ denotes the mean of points in $S$.
\end{definition}

\begin{definition}[Centroid property]
A distance measure $D$ over space $\mathcal{X}$ is said to satisfy the centroid property if for any subset $P \subseteq \mathcal{X}$ and any point $c \in \mathcal{X}$, we have:
\[
\sum_{p \in P} D(p, c) = \Delta_{1}(P) + |P| \cdot D(m(P), c),
\]
where $m(P) = \frac{\sum_{p \in P} p}{|P|}$ denotes the mean of the points in $P$.
\end{definition}

\begin{definition}[$\alpha$-approximate triangle inequality]
Given $\alpha \geq 1$, a distance measure $D$ over space $\mathcal{X}$ is said to satisfy $\alpha$-approximate triangle inequality if for any three points $p, q, r \in \mathcal{X},
D(p, q) \leq \alpha \cdot (D(p, r) + D(r, q))$
\end{definition}

\begin{definition}[$\beta$-approximate symmetry]
Given $0 < \beta \leq 1$, a distance measure $D$ over space $\mathcal{X}$ is said to satisfy $\beta$-symmetric property if for any pair of points $p, q \in \mathcal{X}$, $\beta \cdot D(q, p) \leq D(p, q) \leq \frac{1}{\beta} \cdot D(q, p) $
\end{definition}

The next theorem gives the generalization of our results for distance measures that satisfy the above basic properties. The proof of this theorem follows easily  from the analysis in the previous section. 
The proof of this theorem is given in Appendix~\ref{appendix:A}.

\begin{theorem}\label{thm:other}
Let $f:\mathbb{R} \times \mathbb{R} \rightarrow \mathbb{R}$. Let $\alpha \geq 0$, $0 < \beta \leq 1$, and $0 < \delta < 1/2$ be constants and let $0 < \eps \leq 1/2$. Let $\eta = \frac{2\alpha^2}{\beta^2}(1 + 1/\beta)$. Let $D$ be a distance measure over space $\mathcal{X}$ that $D$ follows:
\begin{enumerate}
\item $\beta$-approximate symmetry property,

\item $\alpha$-approximate triangle inequality,

\item Centroid property, and

\item $(f, \epsilon, \delta)$-sampling property.
\end{enumerate}
Then there is an algorithm that runs in time $O\left(nd  \cdot 2^{\tilde{O}(k \cdot f(\epsilon/\eta k, 0.2))}\right)$ and gives a $(1 + \epsilon)$-approximation to the $k$-median objective for any point set $P \subseteq \mathcal{X}, |P| = n$.
\end{theorem}

The above theorem gives a characterization for when our non-uniform sampling based algorithm can be used to obtain a PTAS for a dissimilarity measure. The important question now is whether there exist interesting distance measures that satisfy the properties in the above Theorem. Next, we look at some distance measures other than squared Euclidean distance, that satisfy such properties.

\subsection{Mahalanobis distance}

Here the domain is $\mathbb{R}^d$ and the distance is defined with respect to a positive definite matrix $A \in \mathbb{R}^{d \times d}$. The distance between two points $p, q \in \mathbb{R}^d$ is given by $D_{A}(p, q) = (p - q)^{T} \cdot A \cdot (p-q)$. Now, we discuss the properties in Theorem~\ref{thm:other}.
\begin{enumerate}
\item ({\it Symmetry}) For any pair of points $p, q \in \mathbb{R}^d$, we have $D_A(p, q) = D_{A}(q, p)$. So, the $\beta$-approximate symmetry property holds for $\beta = 1$.

\item ({\it Triangle inequality}) \cite{ab09} shows that $\alpha$-approximate triangle inequality holds for $\alpha = 2$.

\item ({\it Centroid}) The centroid property is shown to hold for Mahalanobis distance in \cite{ban05}.

\item ({\it Sampling}) \cite{abs10} (see Corollary 3.7) show that Mahalanobis distance satisfy the $(f, \gamma, \delta)$-sampling property for $f(\gamma, \delta) = 1/(\gamma \delta)$.
\end{enumerate}

\noindent
Using the above properties and Theorem~\ref{thm:other}, we get the following result.

\begin{theorem}[$k$-median w.r.t. Mahalanobis distance]
Let $0 < \eps \leq 1/2$. There is an algorithm that runs in time $O(nd \cdot 2^{\tilde{O}(k^2/\eps)})$ and gives a $(1+\eps)$-approximation to the $k$-median objective function w.r.t. Mahalanobis distances for any point set $P \in \mathbb{R}^d, |P| = n$.
\end{theorem}

\subsection{$\mu$-similar Bregman divergence}

We start by defining Bregman divergence and then discuss the required properties.

\begin{definition}[Bregman Divergence]
Let $\phi : X \rightarrow \mathbb{R}^d$ be a continuously-differentiable real-valued and strictly convex function defined on a closed convex set $X$.The Bregman distance associated with $\phi$ for points $p,q \in X$  is:
\begin{eqnarray*}
D_{\phi}(p,q) = \phi(p) - \phi(q)  - \Delta \phi(q)^T(p-q)
\end{eqnarray*}
Where $\Delta \phi(q)$ denotes the gradient of $\phi$ at point $q$
\end{definition}

Intuitively this can be thought of as the difference between the value of $\phi$ at point $p$ and the value of the first-order Taylor expansion of $\phi$ around point $q$ evaluated at point $p$. Bregman divergence includes the following popular distance measures:
\begin{itemize}
\item {\em Euclidean distance.} $D_{\phi}(p, q) = ||p-q||^2$. Here $\phi(x) = ||x||^2$.

\item {\em Kullback-Leibler divergence.} $D_{\phi}(p, q) = \sum_{i} p_i \cdot \ln{\frac{p_i}{q_i}} - \sum_{i}(p_i - q_i)$. Here $D_{\phi}(x) = \sum_{i} x_i \cdot \ln{x_i} - x_i$.

\item {\em Itakura-Saito divergence.} $D_{\phi}(p, q) = \sum_i \left(\ln{\frac{p_i}{q_i}}  - \ln{\frac{q_i}{p_i}} - 1\right)$. Here $\phi(x) = - \sum_{i} \ln{x_i}$.

\item {\em Mahalanobis distance.} For a symmetric positive definite matrix $U \in \mathbb{R}^{d \times d}$, the Mahalanobis distance is defined as:
$D_U(p,q) = (p-q)^TU(p-q)$.
Here $\phi_{U}(x) = x^{T} U x$.
\end{itemize}

Bregman divergences have been shown to satisfy the Centroid property by Banerjee et. al. \cite{ban05}. All Bregman divergences do not necessarily satisfy the symmetry property or the triangle inequality. So, we cannot hope to use our results for the class of all Bregman divergences. On the other hand, some of the Bregman divergences that are used in practice satisfy a property called {\em $\mu$-similarity} (see \cite{a09} for an overview of such Bregman divergences). Next, we give the definition of $\mu$-similarity.

\begin{definition}[$\mu$-similar Bregman divergence]
A Bregman divergence $D_{\phi}$ on domain $\mathbb{X} \subseteq \mathbb{R}^d$ is called $\mu$-similar for constant $0 < \mu \leq 1$, if there exists a symmetric positive definite matrix $U$ such that for Mahalanobis distance $D_U$ and for each $p,q \in \mathbb{X}$ we have:
\begin{equation}\label{eqn:similar}
\mu \cdot D_U(p,q) \leq D_{\phi}(p,q) \leq D_U(p,q).
\end{equation}
\end{definition}

Now, a $\mu$-similar Bregman divergence can easily be shown to satisfy approximate symmetry and triangle inequality properties. This is formalized in the following simple lemma. The proof of this lemma is given in the Appendix~\ref{appendix:B}.

\begin{lemma}\label{lemma:mu-similar}
Let $0 < \mu \leq 1$. Any $\mu$-similar Bregman divergence satisfies the $\mu$-approximate symmetry property and $(2/\mu)$-approximate triangle inequality.
\end{lemma}

Finally, we use the sampling property from Ackermann et. al. \cite{abs10} who show that any $\mu$-similar Bregman divergence satisfy the $(f, \gamma, \delta)$-sampling property for $f(\gamma, \delta) = \frac{1}{\mu \gamma \delta}$.

Using all the results mentioned above we get the following Theorem for $\mu$-similar Bregman divergences.

\begin{theorem}[$k$-median w.r.t. $\mu$-similar Bregman divergences]
Let $0 < \mu \leq 1$ and $0 < \eps \leq 1/2$. There is an algorithm that runs in time $O\left( nd \cdot 2^{\tilde{O}\left(\frac{k^2}{\mu \cdot \eps} \right)} \right)$ and gives a $(1+\eps)$-approximation to the $k$-median objective function w.r.t. $\mu$-similar Bregman divergence for any point set $P \in \mathbb{R}^d, |P| = n$.
\end{theorem}

\bibliography{paper}

\appendix

\section{Proof of Theorem~\ref{thm:other}}\label{appendix:A}
Here we give a proof of Theorem~\ref{thm:other}. For the proof, we repeat the analysis in Section~\ref{sec:algo} almost word-by-word. One the main things we will be doing here is replacing all instances of $d(p, q)^2$ in Section~\ref{sec:algo} with $D(p, q)$. So, this section will look very similar to Section~\ref{sec:algo}. 
First we will restate Theorem~\ref{thm:other}.

\begin{thm}[Restatement of Theorem~\ref{thm:other}]
Let $f:\mathbb{R} \times \mathbb{R} \rightarrow \mathbb{R}$. Let $\alpha \geq 0$, $0 < \beta \leq 1$, and $0 < \delta < 1/2$ be constants and let $0 < \eps \leq 1/2$. Let $\eta = \frac{2\alpha^2}{\beta^2} (1 + 1/\beta)$. Let $D$ be a distance measure over space $\mathcal{X}$ that $D$ follows:
\begin{enumerate}
\item $\beta$-approximate symmetry property,

\item $\alpha$-approximate triangle inequality,

\item Centroid property, and

\item $(f, \epsilon, \delta)$-sampling property.
\end{enumerate}
Then there is an algorithm that runs in time $O\left(nd  \cdot 2^{\tilde{O}(k \cdot f(\epsilon/\eta k, 0.2))}\right)$ and gives a $(1 + \epsilon)$-approximation to the $k$-median objective for any point set $P \subseteq \mathcal{X}, |P| = n$.
\end{thm}

We will first assume that the instance is $(k, \eps)$-irreducible for a suitably small parameter $\eps$. 
We shall then get rid of this assumption later as we did in Section~\ref{sec:algo}. 
The algorithm remains the same and is described in Figure~\ref{fig:k-repeat}. 

We develop some notation first. For the rest of the analysis, we will fix a tuple $(s_1, \ldots, s_k)$ --
this will be the ``desired tuple'', i.e., the one for which we can show that the set $C$ gives a good solution.
As our analysis proceeds, we will argue what properties this tuple should have. Let $\C{i}$ be the set $C$
at the beginning of the $i^{th}$ iteration of Step 2(b). To begin with $\C{0}$ is empty. Let $\itS{i}$
be the set $S$ sampled during the $i^{th}$ iteration of Step 2(b), and $\itT{i}$ be the corresponding set $T$
(which is the $s_i^{th}$ subset of $\itS{i}$).

Let $O_1, \ldots, O_k$ be the optimal clusters, and $c_1,...,c_k$ denote the respective optimal cluster centers. 
Further, let $m_i$ denote $|O_i|$, and wlog assume that $m_1 \geq \ldots \geq m_k$. 
Let $r_i$ denote the average cost paid by a point in $O_i$, i.e.,
$$r_i = \frac{\sum_{p \in O_i} D(p,c_i)}{m_i}. $$

\begin{center}
\begin{figure}
\begin{boxedminipage}{5in}
{\bf Find-k-median(P)}

\hspace{0.1in} Let $\eta = \frac{2 \alpha^2}{\beta^2}(1 + 1/\beta)$, $N = \frac{(24 \eta \alpha \beta k) \cdot f(\eps/\eta, 0.2)}{\eps^2}$, $M = f(\eps/\eta, 0.2)$, and $P = \binom{N}{M}$

\hspace{0.1in} 1. {\bf Repeat} $2^k$ times and output the the set of centers $C$ that give least cost

\hspace{0.3in} 2. {\bf Repeat} for all $k$-tuples $(s_1, ..., s_k) \in [P]\times [P] \times .... \times [P]$ and

\hspace{0.3in} \ \ \ \ \ pick the set of centers $C$ that gives least cost

\hspace{0.5in} \ \ \ (a) $C \leftarrow \{\}$

\hspace{0.5in} \ \ \ (b) For $i \leftarrow 1$ to $k$

\hspace{0.7in}\ \ \ \ \  \ \ Sample a set $S$ of $N$ points with $D^2$-sampling (w.r.t. centers $C$)

\hspace{0.7in} \ \ \ \ \ \ \ Let $T$ be the $s_i^{th}$ subset of $S$. \footnote{For a set of size $N$ we consider an arbitrary ordering of the subsets of size $M$ of this set.}

\hspace{0.7in} \ \ \ \ \ \ \ $C \leftarrow C \cup \{m(T)\}$. \footnote{$m(T)$ denote the centroid of the points in $T$.}

\end{boxedminipage}
\caption{The algorithm that gives $(1+\eps)$-approximation for any $(k, \eps)$-irreducible data set. 
Note that the inner loop is executed at most  $2^k \cdot \left( \binom{N}{M} \right)^k \sim 2^k \cdot 2^{{\tilde O} (k \cdot f(\eps/\eta, 0.2))}$ times. }
\label{fig:k-repeat}
\end{figure}
\end{center}

First, we show that any two optimal centers are far enough.
\begin{lemma}
For any $1 \leq i< j \leq k$,
$$ D(c_j, c_i) \geq \eps \cdot (r_i + r_j). $$
\end{lemma}
\begin{proof}
Since $i < j$, we have $m_i \geq m_j$.
For the sake of contradiction assume $D(c_j, c_i) <  \eps \cdot (r_i + r_j)$. Then we have,
\begin{eqnarray*}
\Delta(O_i \cup O_j, \{c_i\}) & = & m_i \cdot r_i + m_j \cdot r_j + m_j \cdot D(c_j, c_i)  \quad \textrm{(using Centroid property)}\\
&<&  m_i \cdot r_i + m_j \cdot r_j + m_j \cdot \eps \cdot (r_i + r_j)\\
&\leq& (1 + \eps) \cdot m_i \cdot r_i +(1 + \eps) \cdot  m_j \cdot r_j \quad \textrm{(since $m_i \geq m_j$)}\\
&\leq& (1 + \eps) \cdot \Delta(O_i \cup O_j, \{c_i, c_j\})
\end{eqnarray*}
This implies that the centers $\{c_1, ..., c_k\} \backslash \{c_j\}$ give a $(1 + \eps)$-approximation to the
$k$-median objective. This contradicts the assumption that $P$ is $(\eps, k)$-irreducible.
\end{proof}

The above lemma gives the following Corollary that we will use in the rest of the proof.

\begin{corollary}\label{lem:min-repeat}
For any $i \neq j$, $D(c_i, c_j) \geq (\beta \eps) \cdot (r_i + r_j)$.
\end{corollary}
\begin{proof}
If $i > j$, then we have $D(c_i, c_j) \geq \eps \cdot (r_i + r_j)$ from the above lemma and hence $D(c_i, c_j) \geq (\beta \eps) \cdot (r_i + r_j)$. In case $i < j$, then the above lemma gives $D(c_j, c_i) \geq \eps \cdot (r_i + r_j)$. Using $\beta$-approximate symmetry property we get the statement of the corollary.
\end{proof}

We give an outline of the proof. Suppose in the first $(i-1)$ iterations, we have found centers which are close to the centers of some $(i-1)$ clusters in the optimal solution. Conditioned on this fact, we show that in the next iteration, we are likely to sample enough number of points from one of the remaining clusters (c.f. Corollary~\ref{cor:sample-repeat}). Further, we show that the samples from this new cluster are close to uniform distribution (c.f. Lemma~\ref{lem:key-repeat}). Since such a sample does not come from exactly uniform distribution, we cannot use the $(f, \gamma, \delta)$-sampling property directly. 
In fact, dealing with the slight non-uniformity turns out to be non-trivial (c.f. Lemmas~\ref{lem:clinaba-repeat} and ~\ref{lem:final-repeat}). 

We now show that the following invariant will hold for all iterations  : let $\C{i-1}$ consist of centers
$c_1', \ldots, c_{i-1}'$ (added in this order). Then, with probability at least $\frac{1}{2^{i}}$, there exist distinct
indices $j_1, \ldots, j_{i-1}$ such that
for all $l = 1, \ldots, i-1$,
\begin{equation}
\label{eq:inv-repeat}
 \Delta(O_{j_l},c_l') \leq (1 + \eps/\eta) \cdot \Delta(O_{j_l},c_{j_l})
\end{equation}
Where $\eta$ is a fixed constant that depends on $\alpha$ and $\beta$. With foresight, we fix the value of $\eta = \frac{2 \alpha^2}{\beta^2} \cdot (1 + 1/\beta)$.
Suppose this invariant holds for $\C{i-1}$ (the base case is easy since $\C{0}$ is empty).
We now show that this invariant holds for $\C{i}$ as well. In other words, we just show that in the $i^{th}$ iteration,
with probability at least $1/2$,  the
algorithm finds a center $c_{i}'$  such that
$$ \Delta(O_{j_i},c_i') \leq (1 + \eps/\eta) \cdot \Delta(O_{j_i},c_{j_i}), $$ where $j_i$ is an index distinct from
$\{j_1, \ldots, j_{i-1}\}$.
This will basically show that at the end of the last iteration, we will have $k$ centers that give a $(1 + \eps)$-approximation
with probability at least $2^{- k}$.

We now show that the invariant holds for $\C{i}$. We use the notation developed above for $\C{i-1}$. Let $I$ denote
the set of indices $\{j_1, \ldots, j_{i-1}\}$. Now let $j_i$ be the index $j \notin I$ for which $\Delta(O_j, \C{i-1})$
is maximum. Intuitively, conditioned on sampling from clusters in $O_{i}, \cdots, O_{k}$ using $D^2$-sampling, it is  likely that
enough points from $O_{j_i}$ will be sampled.
The next lemma shows that there is good chance that elements from the sets $O_j$ for $j \notin I$
will be sampled.

\begin{lemma}
\label{lem:sampled-repeat}
 $$\frac{\sum_{l \notin I} \Delta(O_l, \C{i-1})}{\sum_{l=1}^k \Delta(O_l, \C{i-1})} \geq \eps/2.$$
\end{lemma}
\begin{proof}
Suppose, for the sake of contradiction, the above statement does not hold. Then,
\begin{eqnarray*}
\Delta(P, \C{i-1}) &=& \sum_{l \in I}  \Delta(O_l, \C{i-1}) +  \sum_{l \notin I} \Delta(O_l, \C{i-1})  \\
& < &  \sum_{l \in I} \Delta(O_l, \C{i-1})  +  \frac{\eps/2}{1 - \eps/2} \cdot \sum_{l \in I}  \Delta(O_l, \C{i-1})
\quad \textrm{(by our assumption)}\\
&=& \frac{1}{1 - \eps/2} \cdot \sum_{l \in I} \Delta(O_l, \C{i-1}) \\
&\leq& \frac{1 + \eps/\eta}{1 - \eps/2} \cdot \sum_{l \in I}  \Delta_1(O_l) \quad \textrm{(using the invariant for $\C{i-1}$)}\\
&\leq& (1+\eps) \cdot  \sum_{l \in I}  \Delta_1(O_l) \quad \textrm{(using  $\eta = (2\alpha^2/\beta^2)\cdot (1 + 1/\beta) \geq 4$)}\\
&\leq& (1+\eps) \cdot  \sum_{l \in [k]}  \Delta_1(O_l)
\end{eqnarray*}

But this contradicts the fact that $P$ is $(k, \eps)$-irreducible.
\end{proof}

\noindent
We get the following corollary easily.
\begin{corollary}
\label{cor:sample-repeat}
$$ \frac{\Delta(O_{j_i}, \C{i-1})}{\sum_{l=1}^k  \Delta(O_l, \C{i-1})} \geq \frac{\eps}{2k}.$$
\end{corollary}

The above Lemma and its Corollary say that with probability at least $\frac{\eps}{2k}$,
points in the set $O_{j_i}$ will be sampled. However the points within $O_{j_i}$ are not sampled uniformly.
Some points in $O_{j_i}$ might be sampled with higher probability than other points.
In the next lemma, we show that each point will be sampled with certain minimum probability.

\begin{lemma}
\label{lem:key-repeat}
For any $l \notin  I$ and any point $p \in O_l$, $$\frac{D(p, \C{i-1})}{\Delta(O_l, \C{i-1})} \geq \frac{1}{m_l} \cdot \frac{\eps}{3 \alpha \beta \eta}. $$
\end{lemma}
\begin{proof}
Fix a point $p \in O_l$.
Let $j_t \in I$ be the
 index such that $p$ is closest to $c_{t}'$ among all centers in $\C{i-1}$.
We have
\begin{eqnarray}
\Delta(O_{l}, \C{i-1}) &\leq& m_{l} \cdot r_{l} + m_l \cdot D(c_l, c_t') \quad \textrm{(using Centroid property)}\nonumber \\
&\leq& m_{l} \cdot r_{l} + \alpha \cdot m_{l} \cdot \left(D(c_l, c_{j_t}) + D(c_{j_t}, c_t') \right) \quad \textrm{(Using  triangle inequality)} \nonumber \\
\label{eq:p1-repeat}
&\leq& m_l \cdot r_l + \alpha \cdot m_l \cdot \left( D(c_l, c_{j_t}) + \frac{ \eps r_{j_t}}{\eta} \right),
\end{eqnarray}
where the last inequality follows from the invariant condition for $\C{i-1}$.
Also, we know that the following inequalities hold:
\begin{equation}\label{u1}
\alpha \cdot (D(p, c_t') + D(c_t', c_{j_t})) \geq D(p, c_{j_t}) \quad \textrm{(from approximate triangle inequality)}
\end{equation}
\begin{equation}\label{u2}
\alpha \cdot (D(c_l, p) + D(p, c_{j_t})) \geq D(c_l, c_{j_t}) \quad \textrm{(from approximate triangle inequality)}
\end{equation}
\begin{equation}\label{u3}
D(p, c_l) \leq D(p, c_{j_t}) \quad \textrm{(since $p \in O_l$)}
\end{equation}
\begin{equation}\label{u4}
\beta \cdot D(c_l, p) \leq D(p, c_l) \leq (1/\beta) \cdot D(c_l, p) \quad \textrm{(from approximate symmetry)}
\end{equation}
\begin{equation}\label{u5}
D(c_{j_t}, c_t') \leq (\eps/\eta) \cdot r_{j_t} \quad \textrm{(from invariant condition)}
\end{equation}
\begin{equation}\label{u6}
\beta \cdot D(c_{j_t}, c_t') \leq D(c_t', c_{j_t}) \leq (1/\beta) \cdot D(c_{j_t}, c_t')  \quad \textrm{(from approximate symmetry)}
\end{equation}
Inequalities (\ref{u2}), (\ref{u3}), and (\ref{u4}) gives the following:
\begin{eqnarray}
&& D(p, c_{j_t})  + D(c_l, p) \geq \frac{D(c_l, c_{j_t})}{\alpha}  \nonumber \\
&& \Rightarrow D(p, c_{j_t})  + \frac{D(p, c_l)}{\beta} \geq \frac{D(c_l, c_{j_t})}{\alpha}\quad \textrm{(using (\ref{u4}))}\nonumber \\
&& \Rightarrow  D(p, c_{j_t}) + \frac{D(p, c_{j_t})}{\beta} \geq \frac{D(c_l, c_{j_t})}{\alpha} \quad \textrm{(using (\ref{u3}))}\nonumber \\
\label{u7}
&& \Rightarrow D(p, c_{j_t}) \geq \frac{D(c_l, c_{j_t})}{\alpha (1 + 1/\beta)} 
\end{eqnarray}
Using (\ref{u1}) and (\ref{u7}), we get the following:
\begin{equation*}
D(p, c_t') \geq \frac{D(c_l, c_{j_t})}{\alpha^2 (1 + 1/\beta)} - D(c_t', c_{j_t})
\end{equation*}
Using the previous inequality and (\ref{u6}) we get the following:
\begin{eqnarray}
D(p,c_t') & \geq & \frac{D(c_l, c_{j_t})}{\alpha^2 (1 + 1/\beta)} - \frac{D(c_{j_t}, c_t')}{\beta} \nonumber \\
& \geq & \frac{D(c_l, c_{j_t})}{\alpha^2 (1+1/\beta)} - \frac{\eps}{\eta \beta} \cdot r_{j_t} \quad \textrm{(using the invariant for $\C{i-1}$)} \nonumber \\
\label{eq:p2-repeat}
&\geq& \frac{D(c_l, c_{j_t})}{\eta \beta^2} \quad \textrm{(Using Corollary~\ref{lem:min-repeat})}
\end{eqnarray}
So, we get
\begin{eqnarray*}
\frac{D(p, \C{i-1})}{\Delta(O_l, \C{i-1})} & \geq& \frac{D(c_l, c_{j_t})}
{(\eta \beta^2) \cdot m_l \cdot \left( r_l + \alpha \left( D(c_l, c_{j_t}) + \frac{\eps r_t}{\eta} \right) \right)} \quad \textrm{(using (\ref{eq:p1-repeat}) and (\ref{eq:p2-repeat}))} \\
&\geq& \frac{1}{(\eta \beta^2) \cdot m_l} \cdot \frac{1}{1/(\beta \eps) + \alpha + 1/(\eta \beta)} \\
&\geq& \frac{\eps}{(3\eta \alpha \beta)} \cdot \frac{1}{m_l} \quad \textrm{(using Corollary~\ref{lem:min-repeat})}
\end{eqnarray*}
\end{proof}

Recall that $\itS{i}$ is the sample of size $N$ in this iteration. We would like to show that  that the invariant
will hold in this iteration as well. We first prove a simple corollary of Lemma~\ref{lem:inaba}.

\begin{lemma}
\label{lem:clinaba-repeat}
Let $Q$ be a set of $n$ points, and $\gamma$ be a parameter, $0 < \gamma < 1$. Define a random variable $X$ as follows :
with probability $\gamma$, it picks an element of $Q$ uniformly at random, and with  probability $1-\gamma$, it does not
pick any element (i.e., is null). Let $X_1, \ldots, X_\ell$ be $\ell$ independent copies of $X$, where $\ell = \frac{4}{\gamma} \cdot f(\eps/\eta, 0.2).$
Let $T$ denote the (multi-set) of elements of $Q$ picked by $X_1, \ldots, X_\ell$. Then, with probability at least $3/4$,
$T$ contains a subset $U$ of size $f(\eps/\eta, 0.2)$ which satsifies
\begin{eqnarray}
\label{eq:clinaba-repeat}
\Delta(P, m(U)) \leq \left(1 + \frac{\eps}{\eta} \right) \cdot \Delta_1(P)
\end{eqnarray}
\end{lemma}
\begin{proof}
 Define a random variable $I$, which is a subset of the index set $\{1, \ldots, \ell\}$, as follows
$I = \{ t : X_t \mbox{ picks an element of $Q$, i.e., it is not null} \}$. Conditioned on $I = \{t_1, \ldots, t_r\}$, note that the random variables $X_{t_1}, \ldots, X_{t_r}$ are independent uniform samples from $Q$. 
Thus if $|I| \geq f(\eps/\eta, 0.2)$, then sampling property wrt. $D$ implies that with 
probability at least 0.8, the desired event~(\ref{eq:clinaba-repeat}) happens. 
But the expected value of $|I|$ is $4 \cdot f(\eps/\eta, 0.2)$, and so, $|I| \geq f(\eps/\eta, 0.2)$ with high probability, and hence, the statement in the lemma is true.
\end{proof}

\noindent
We are now ready to prove the main lemma.
\begin{lemma}
\label{lem:final-repeat}
With probability at least $1/2$, there exists a subset $\itT{i}$ of $\itS{i}$ of size at most $f(\eps/\eta, 0.2)$ such that
$$ \Delta(O_{j_i}, m(\itT{i})) \leq \left(1 + \frac{\eps}{\eta}\right) \cdot \Delta_1(O_{j_i}). $$
\end{lemma}
\begin{proof}
Recall that $\itS{i}$ contains $N = \frac{(24 \eta \alpha \beta k) \cdot f(\eps/\eta, 0.2)}{\eps^2}$ independent samples of $P$ (using $D^2$-sampling). We are interested in $\itS{i} \cap O_{j_i}$.
Let $Y_1, \ldots, Y_N$ be  $N$ independent random variables defined as follows : for any $t$, $1 \leq t \leq N$, $Y_t$ picks an element of $P$ using $D^2$-sampling with respect to $\C{i-1}$. 
If this element is not in $O_{j_i}$, it just discards it (i.e., $Y_t$ is null). 
Let $\gamma$ denote $\frac{\eps^2}{6 \eta \alpha \beta k}$. Corollary~\ref{cor:sample-repeat} and Lemma~\ref{lem:key-repeat} imply that $Y_t$ picks a particular element of $O_{j_i}$ with probability at least $\frac{\gamma}{m_{j_i}}$. 
We would now like to apply Lemma~\ref{lem:clinaba-repeat} (observe that $N = \frac{4}{\gamma} \cdot f(\eps/\eta, 0.2)$). 
We can do this by a simple coupling argument as follows.
For a particular element $p \in O_{j_i}$, suppose $Y_t$ assigns probability $\frac{\gamma(p)}{m_{j_i}}$ to it.
One way of sampling a random variable $X_t$ as in Lemma~\ref{lem:clinaba-repeat} is as follows --  first sample using $Y_t$. If $Y_t$ is null, then $X_t$ is also null. Otherwise, suppose $Y_t$ picks an element $p$ of $O_{j_i}$. 
Then $X_t$ is equal to $p$ with probability $\frac{\gamma}{\gamma(p)}$, and null otherwise. 
It is easy to check that with probability $\gamma$, $X_t$ is a uniform sample from $O_{j_i}$, and null with probability $1-\gamma$. 
Now, observe that the set of elements of $O_{j_i}$ sampled by $Y_1, \ldots, Y_N$ is always a superset of $X_1, \ldots, X_N$. 
We can now use Lemma~\ref{lem:clinaba-repeat} to finish the proof.
\end{proof}

 Thus, we will take the index $s_i$ in Step 2(i) as the index of the set $\itT{i}$ as guaranteed by the Lemma above.
Finally, by repeating the entire process $2^k$ times, we make sure that we get a $(1+\eps)$-approximate solution
with high probability.
Note that the total running time of our algorithm is $\left( nd \cdot 2^k \cdot 2^{\tilde{O}(k \cdot f(\eps/\eta, 0.2))} \right)$.

\noindent
{\bf Removing the $(k,\eps)$-irreducibility assumption :} We now show how to remove this assumption. First note
that we have shown the following result.

\begin{theorem}
If a given point set $(k, \frac{\eps}{(1+\eps/2) \cdot k})$-irreducible, then
there is an algorithm that gives a $(1 + \frac{\eps}{(1+\eps/2) \cdot k})$-approximation to the $k$-median objective
with respect to distance measure $D$ and that runs in time $O(nd \cdot 2^{\tilde{O}(k \cdot f(\eps/k\eta, 0.2))})$.
\end{theorem}
\begin{proof}
The proof can be obtained by replacing $\eps$ by $\frac{\eps}{(1 + \eps/2) \cdot k}$ in the above analysis. \end{proof}

Suppose the point set $P$ is not $(k,\frac{\eps}{(1+\eps/2) \cdot k})$-irreducible.
In that case it will be sufficient to find fewer centers that $(1 + \eps)$-approximate the $k$-median objective.
The next lemma shows this more formally.

\begin{theorem}
There is an algorithm that runs in time $O(nd \cdot 2^{\tilde{O}(k \cdot f(\eps/\eta k, 0.2))})$ and
gives a $(1 + \eps)$-approximation to the $k$-median objective with respect to $D$.
\end{theorem}
\begin{proof}
Let $P$ denote the set of points.  Let $1 < j \leq k$ be the largest index such that $P$ is $(i, \frac{\eps}{(1+\eps/2) \cdot k})$-irreducible.
If no such $i$ exists, then $$\Delta_1(P) \leq \left(1+\frac{\eps}{(1+\eps/2)\cdot k} \right)^k \cdot \Delta_k(P) \leq (1+\eps) \cdot \Delta_k(P), $$
and so picking the centroid of $P$ will give a $(1+\eps)$-approximation.

Suppose such an $i$ exists. In that case, we consider the $i$-median problem and from the previous lemma
we get that there is an algorithm that runs in time $O(nd \cdot 2^i \cdot 2^{\tilde{O}(i \cdot f(\eps/\eta k, 0.2))})$ 
and gives a
$(1 + \frac{\eps}{(1+\eps/2) \cdot k})$-approximation to the $i$-median objective.
Now we have that $$\Delta_i \leq \left(1 + \frac{\eps}{(1+\eps/2) \cdot k} \right)^{k-i} \cdot \Delta_k \leq (1+\eps) \cdot \Delta_k.$$
Thus, we are done.
\end{proof}

\section{Proof of Lemma~\ref{lemma:mu-similar}}\label{appendix:B}

Here we give the proof of Lemma~\ref{lemma:mu-similar}. For better readability, we first restate the Lemma.
\begin{lm}[Restatement of Lemma~\ref{lemma:mu-similar}]
Let $0 < \mu \leq 1$. Any $\mu$-similar Bregman divergence satisfies the $\mu$-approximate symmetry property and $(2/\mu)$-approximate triangle inequality.
\end{lm}

The above lemma follows from the next lwo sub-lemmas.

\begin{lemma}[Symmetry for $\mu$-similar Bregman divergence]
Let $ 0 < \mu \leq 1$. Consider a $\mu$-similar Bregman divergence $D_{\phi}$ on domain $\mathbb{X} \subseteq \mathbb{R}^d$. For any two points $p, q \in \mathbb{X}$, we have:
$
\mu \cdot D_{\phi}(q, p) \leq D_{\phi}(p, q) \leq \frac{1}{\mu} \cdot D_{\phi}(q, p)
$
\end{lemma}
\begin{proof}
Using equation(\ref{eqn:similar}) we get the following:
\[
\mathbf{\mu \cdot D_{\phi}(q, p)} \leq \mu \cdot D_{U}(q, p) = \mu \cdot D_{U}(p, q) \leq  \mathbf{D_{\phi}(p, q)} \leq D_{U}(p, q) = D_{U}(q, p) \leq \mathbf{\frac{1}{\mu} \cdot D_{\phi}(q, p)}.
\]
\end{proof}

\begin{lemma}[Triangle inequality for $\mu$-similar Bregman divergence]
Let $ 0 < \mu \leq 1$. Consider a $\mu$-similar Bregman divergence $D_{\phi}$ on domain $\mathbb{X} \subseteq \mathbb{R}^d$. For any three points $p, q, r \in \mathbb{X}$, we have:
$
(\mu/2) \cdot D_{\phi}(p, r) \leq D_{\phi}(p, q) + D_{\phi}(q, r) 
$
\end{lemma}
\begin{proof}
We have:
\begin{eqnarray*}
D_{\phi}(p, q) + D_{\phi}(q, r) &\geq& \mu \cdot (D_{U}(p, q) + D_{U}(q, r)) \\
&\geq& (\mu/2) \cdot D_{U}(p, r)\\
&\geq& (\mu/2) \cdot D_{\phi}(p, r)
\end{eqnarray*}
The first and third inequality is using equation~\ref{eqn:similar} and the second inequality is using the approximate triangle inequality for Mahalanobis distance.
\end{proof}
\end{document}